\documentclass{llncs}
\usepackage{makeidx}  

\usepackage{makeidx}  
\usepackage{amsmath,amssymb,amscd,oldlfont,color}

\usepackage{xspace}
\usepackage{times}
\usepackage{proof}
\usepackage{graphicx}
\usepackage{proof,colordvi}

\newcommand{\out}[1]{}
\newcommand{\delete}[1]{}
\newcommand{\notedomission}[1]{\medskip\noindent{\bf TEXT OMITTED}\\[2mm]}

\newcommand{\ignore}[1]{}

\newcommand{\bc}{\begin{center}}
\newcommand{\ec}{\end{center}}
\newcommand{\beq}{\begin{equation}}
\newcommand{\eeq}{\end{equation}}
\newcommand{\be}{\begin{enumerate}}
\newcommand{\ee}{\end{enumerate}}
\newcommand{\bi}{\begin{itemize}}
\newcommand{\ei}{\end{itemize}}
\newcommand{\bd}{\begin{description}}
\newcommand{\ed}{\end{description}}
\newcommand{\beqn}{\begin{equation}}
\newcommand{\eeqn}{\end{equation}}

\newcommand{\beqna}{\begin{eqnarray}}
\newcommand{\eeqna}{\end{eqnarray}}
\newcommand{\beqnas}{\begin{eqnarray*}}
\newcommand{\eeqnas}{\end{eqnarray*}}
\newcommand{\beqnaa}{$$\begin{array}{rcll}}  
\newcommand{\eeqnaa}{\end{array}$$}  
\newcommand{\beqnana}{$$\begin{array}{lrcll}}  
\newcommand{\eeqnana}{\end{array}$$}  
\newcommand{\btbl}[1]{\begin{center}\begin{tabular}{#1}}
\newcommand{\etbl}{\end{tabular}\end{center}}
\newcommand{\beqnc}{$$\begin{array}{rclcl}}
\newcommand{\eeqnc}{\end{array}$$}

\newenvironment{comment}{\begin{quote}\small}{\end{quote}}

\newtheorem{dclprop}{{\sc Proposition}} 
\newtheorem{dclprops}{{\sc Proposition}}[subsection] 
\newtheorem{dclbigthm}[dclprop]{THEOREM}

\makeatletter
\def\thmlabel#1{\@bsphack\if@filesw {\let\thepage\relax
\xdef\@gtempa{\write\@auxout{\string
\newlabel{#1}{{\@Roman{\@currentlabel}}{\thepage}}}}}\@gtempa
\if@nobreak \ifvmode\nobreak\fi\fi\fi\@esphack}
\makeatother

\newtheorem{dclthm}[dclprop]{{\sc Theorem}}   
\newtheorem{dclthms}[dclprops]{{\sc Theorem}}   

\newtheorem{dcllem}[dclprop]{{\sc Lemma}}
\newtheorem{dcllems}[dclprops]{{\sc Lemma}} 
\newtheorem{dclsublem}[dclprop]{{\sc Sublemma}}
\newtheorem{dclcor}[dclprop]{{\sc Corollary}}
\newtheorem{dclcors}[dclprops]{{\sc Corollary}} 
\newtheorem{dcldfn}[dclprop]{{\sc Definition}}
\newtheorem{dcldfns}[dclprops]{{\sc Definition}}
\newtheorem{dclasss}[dclprops]{{\bf Assumption}}
\newtheorem{dclass}[dclprop]{{\bf Assumption}}

\newenvironment{prop}{\medskip\begin{dclprop}\sl}{\end{dclprop}}

\newcommand{\bsl}{\begin{verse}\sl}
\newcommand{\esl}{\end{verse}}

\newtheorem{exxs}[dclprop]{Exercises}

\newenvironment{exercises-with-preamble}{\begin{exxs}\rm}{\end{exxs}}

\newcommand{\bz}{\begin{quote}\small}
\newcommand{\ez}{\end{quote}}

\newcommand{\einference}[2]  
  {\shortstack
      {$ #1 $\\ \mbox{}\\ $ #2 $}}

\newlength{\txtlth}
\newlength{\txtht}

\renewcommand{\phi}{\varphi}
\newcommand{\logspace}{\textsc{logspace}\xspace}
\newcommand{\nlogspace}{\textsc{nlogspace}\xspace}

\newcommand{\node}{\textit{node}}
\newcommand{\startnode}{\textit{startnode}}
\newcommand{\curr}{\textit{curr}}
\newcommand{\grid}{\textit{grid}}
\newcommand{\poly}{\textit{poly}}
\newcommand{\CG}{\textit{CG}}
\newcommand{\targetnode}{\textit{targetnode}}
\newcommand{\target}{\textit{target}}

\newcommand{\purple}{\textsc{purple}\xspace}

\newcommand{\jag}{\textsc{jag}\xspace}
\newcommand{\jags}{\textsc{jag}s\xspace}
\newcommand{\ndjag}{\textsc{nd-jag}\xspace}
\newcommand{\ndjags}{\textsc{nd-jag}s\xspace}

\newcommand{\ramjags}{\textsc{ramjag}s\xspace}

\DeclareMathAlphabet{\mathitbf}{OML}{cmm}{b}{it}

\begin{document}

\mainmatter              

\title{Power of Nondeterministic JAGs on Cayley graphs%
\thanks{This work was supported by Deutsche Forschungsgemeinschaft (\textsc{dfg}) under grant \textsc{purple}.}
}

\author{Martin Hofmann \and Ramyaa Ramyaa
}
\institute{Ludwig-Maximilians Universit\"at M\"unchen\\
Oettingenstra{\ss}e 67, 80538 Munich, Germany\\
\texttt{\{mhofmann,ramyaa\}@tcs.ifi.lmu.de}
}

\authorrunning{M. Hofmann and R. Ramyaa}

%

\maketitle



\begin{abstract} The Immerman-Szelepcsenyi Theorem uses an algorithm
  for co-st-connectivity based on inductive
 counting to prove that
  \nlogspace is closed under complementation. We want to investigate
  whether counting is necessary for this theorem to hold.
   Concretely,  we show that Nondeterministic Jumping Graph Autmata (\ndjags) (pebble automata on graphs),  on   several families of Cayley  graphs, are equal in power to nondeterministic logspace Turing machines that are given such graphs as a linear encoding.
   In particular, it follows that \ndjags can solve
  co-st-connectivity on those graphs.  This came as a surprise since
  Cook and Rackoff showed that deterministic \jags cannot solve
  st-connectivity on many Cayley graphs due to their high
  self-similarity (every neighbourhood looks the same). Thus, our
  results show that on these graphs, nondeterminism provably adds
  computational power.
 
 The families of Cayley graphs 
we consider 
 include Cayley graphs of abelian groups and of all finite
 simple groups irrespective of how they are presented and graphs
 corresponding to groups generated by various product constructions,
 including iterated ones.
 
  We remark that assessing the precise power of nondeterministic
  \jags
 and in particular whether they can solve co-st-connectivity
  on arbitrary graphs is left as an open
 problem by Edmonds, Poon
  and Achlioptas. Our results suggest a positive answer to this
  question and in particular considerably limit the search space for a
  potential counterexample.
 \end{abstract}




\section{Introduction}
The technique of inductive counting forces nondeterministic machines to
enumerate all possible sequences of nondeterministic choices. This is used in
the Immerman-Szelepcsenyi theorem (\cite{Immerman}) proving that
nondeterministic space is closed under complementation and in 
\cite{Szelepcsenyi,smallinductivecounting}) to show similar results
for various space-complexity classes.
This raises the question whether the operation of counting is needed
for such exhaustive enumeration or indeed to show the closure under
complementation of the space classes. To investigate this question for
\nlogspace, we explore the problem of whether counting is needed to
solve co-st-connectivity, i.e., the problem of determining that there
is no path between two specified vertices of a given graph. Note that
for nondeterministic machines such as our Pebble automata this is not
trivially equivalent to st-connectivity (existence of a path) due to
their asymmetric acceptance condition. For nondeterministic Turing
machines with logarithmic space bound whose input (e.g.\ a graph) is
presented in binary encoding on a tape, we know of course by Immerman
Szelepcsenyi's result that the two problems are
equivalent. By \nlogspace we mean as usual the problems solvable by
such a machine; thus, more generally, \nlogspace = co-\nlogspace by
Immerman-Szelepcsenyi.

The framework to study the precise role of counting in that argument should
be a formal system which cannot a priori count, but when augmented with
counting, captures all of  \nlogspace .   There are many
natural systems operating over graphs which provably cannot count,
e.g.\ transitive closure logic (TC) (\cite{reachability-lo}).




These systems get abstract graphs (not their encoding) as input  and perform graph based operations such as moving pebbles along an edge. The main difference between such systems and Turing machines are that the latter get as input graph encodings, which embodies an ordering over the nodes. A total ordering over the nodes of a graph can be used to simulate coutning. The system getting abstract graphs as inputs can still simulate counters up to certain values depending on the structure of the input graph.

 In the case where a graph is presented as a binary relation it has
 been shown that positive TC (first order logic with the transitive closure operator used only in positive places within any formula) cannot solve
 co-st-connectivity \cite{pos-tc}. In the, arguably more natural, case
 where the edges adjacent to a node are ordered no such result is
 known. It is also to this situation that the aforementioned open
 question applies.

We consider here the simplest machines operating on such edge-ordered
graphs - the nondeterministic version of jumping automata on graphs
(\jags).

Cook and Rackoff \cite{jag} introduced jumping automata on graphs 
\jags (a finite automaton which can move a fixed number of pebbles along edges of a
graph) and showed that they cannot solve st-connectivity on undirected graphs (ustcon).
It also known \cite{ramjag} that \jags  augmented with counters, called 
\ramjags,  can simulate Reingold's algorithm \cite{Reingold-ud} for ustcon. Among the consequences of
this are that \jags only need the addition of counters to capture all
of \logspace on connected graphs  and that without the
external addition of counters, \jags cannot count. 
Thus, 
\ndjags may provide us with a platform to explore whether counting is needed for
co-st-connectivity and more generally all of \nlogspace.
 While nondeterministic
\jags (\ndjags)  can obviously solve st-connectivity, even for directed graphs,  it is unclear
whether they can solve co-st-connectivity. 
Note that if an \ndjag, on an input directed graph, can order the graph when the edge directions are disregarded, one could order the underlying undirected graph and use that order to simulate Immerman-Szelepcsényi so as to decide co-st-connectivity for the directed graph. So, nothing is lost by restricting our investigation to undirected graphs. 

Several extensions of \jags  such as 
\jags  working on graphs with named nodes (\cite{nnjags}), probabilistic, randomized and  nondeterministic   
 \jags (\cite{probjag},\cite{nnjags},\cite{nnjag2}) have been studied, but with the motivation of establishing space/time lower or upper bounds of solving reachability questions as measured by the number of states as a function of size and degree of the input graph. This differs from our motivation. 
 
A non-constant lower bound in this sense for solving co-st-connectivity with nd-jags would imply that counting is needed, but we are not aware of space lower bounds for  \ndjags. 

The question of implementing Immerman-Szelepcsenyi's  algorithm 
 in
\ndjags is considered in \cite{nnjag2}, but has been left as an
open question. 

The results of this paper constitute an important step towards
the solution of this question. Namely, we show that \ndjags can
order and thus implement any \nlogspace-algorithm on a variety of
graph families derived from Cayley graphs.  This came as a surprise to
us because Cayley graphs were used by Cook and Rackoff to show that
deterministic \jags cannot solve st-connectivity. Intuitively, this is
due to the high degree of self-similarity of these graphs (every
neighbourhood looks the same).  It thus seemed natural to conjecture,
as we did, that \ndjags cannot systematically explore such graphs
either and tried to prove this. In the course of this attempt we found
the conjecture to be false and discovered that in the case of pebble
automata, nondeterminism indeed adds power.

Our main results show that the Cayley graphs of the following groups
can be ordered by \ndjags: all abelian groups
and simple finite groups irrespective of the choice of generators,
(iterated) wreath products $G\wr H$ once $G$ can be ordered and a mild
size condition on $H$ holds. This implies that none of these groups
can serve as a counterexample for a possible proof that \ndjags cannot
solve co-st-connectivity on undirected graphs.

One may criticize that we do not in this paper answer the question
whether or not \ndjags can solve co-st-connectivity on undirected
graphs. However, this is not at all an easy question and we believe
that our constructions will help to eventually solve it because they
allow one to discard many seemingly reasonable candidates for
counterexamples such as the original counterexamples from Cook and
Rackoff or the iterated wreath products thereof that were used by one
of us and U.\ Sch\"opp \cite{purple-reach} in an extension of Cook and
Rackoff's result.  

An anonymous reviewer of an earlier version wrote
that trying to argue that \ndjags can do co-st-conn without counting
``seems silly''. One should ``just take a graph that is between
structured and unstructured and then it likely can't be done.''. When
we started this work this would have been our immediate reaction, but
on the one hand, it is not easy to construct a graph ``between
structured and unstructured'' in such a way that one can prove
something about it. On the other hand, we found the additional power
of nondeterminism in this context so surprising that we are no longer
convinced that \ndjags really cannot do co-st-connectivity and hope
that our constructions will be helpful in the process of deciding this
question.

\section{Preliminaries/definitions}
\subsection{Jumping Automata on Graphs}
Cook and Rackoff (\cite{jag}) introduced Jumping Automata on Graphs
(\jags) in order to study space lower bounds for reachability problems.
  A \jag is a finite automaton  with a fixed number of pebbles and a transition table. It gets as input a graph of fixed degree and labeled edges and begins its computation with all of its pebbles placed on a distinguished start node. During the course of its computation, the \jag moves the pebbles (a pebble may be moved along an edge of the node it is on, or to the node that has another pebble on it) according to the transition table.  Thus, \jags
 are a nonuniform machine model, with different machines
for graphs of different degrees. This does not affect
reachability results, since graphs of arbitrary degrees can be transformed into 
graphs of a fixed degree while preserving connectivity relations.

A labelled degree $d$ graph for $d>1$ comprises a set $V$ of vertices
and a function $\rho:V\times \{1,\dots,d\}\rightarrow V$. If
$\rho(v,i)=v'$ then we say that $(v,v')$ is an edge labelled $i$ from
$v$ to $v'$. All graphs considered in this paper are labelled degree
$d$ graphs for some $d$. The important difference to the more standard
graphs of the form $G=(V,E)$ where $E\subseteq V\times V$ is that the
out degree of each vertex is exactly $d$ and, more importantly, that
the edges emanating from any one node are linearly ordered. One
extends $\rho$ naturally to sequences of labels (from $\{1,\dots,
d\}^*$) and writes $v'=v.w$ if $\rho(v,w)=v'$ for $v,v'\in V$ and
$w\in\{1,\dots,d\}^*$. The induced sequence of intermediate vertices
(including $v,v'$) is called the path labelled $w$ from $v$ to
$v'$. Such a graph is undirected if for each edge there is one in the
opposite direction $(\rho(v,i)=v'\Rightarrow \exists
j.\rho(v',j)=i$. It is technically useful to slightly generalise this
and also regard such graphs as undirected if each edge can be reversed
by a path of a fixed maximum length.

\begin{definition}
 A $d$-Jumping Automaton for Graphs
 ($d$-\jag),  $J$, consists of
 \begin{itemize}
 \setlength{\itemsep}{1pt}
                \setlength{\parskip}{0pt}
                \setlength{\parsep}{0pt}
                \setlength{\leftmargin}{-.25in}
  \item a finite set $Q$ of states with distinguished start state $q_0$ and accept state $q_a$
  \item a finite set  $P$ of objects called pebbles (numbered $1$ through 
  $p$)
  \item  a transition function $\delta$ which assigns to each state $q$ and each equivalence relation $\pi$ on $P$ (representing incidence of pebbles) a set of pairs $(q',\vec c)$ where $q'\in Q$ is the successor state and where $\vec c=(c_1,\dots,c_p)$ is a sequence of moves, one for each pebble. Such a move can either be of the form $\text{move}(i)$ where $i\in\{1,\dots,d\}$ (move the pebble along edge $i$) or $\text{jump}(j)$ where $j\in\{1,\dots,p\}$ (jump the pebble to the (old) position of pebble $j$). 
\item The automaton is deterministic if $\delta(q,\pi)$ is a singleton set for each $q,\pi$. 
 \end{itemize}
\end{definition}
 
 The input to a  \jag  is a labelled degree $d$ graph. 
 An \emph{instantaneous description} (id) of a \jag $J$
 on an input graph $G$ is specified by a state $q$ and a function, $\node$, from the $P$ to the nodes of $G$ where for any pebble $p$, $\node(p)$ gives the node on which the pebble $p$ is placed. 

Given an id $(q,\node)$ a legal move of $J$ is an element $(q',\vec c)\in \delta(q,\pi)$ where $\pi$ is the equivalence relation given by $p \mathrel{\pi} p'\iff \node(p)=\node(p')$.

The action of a
 \jag,  or the \emph{next move} is given by its transition function and consists of the control passing to a new state after moving each pebble $i$ according to $c_i$: (a) if $c_i=\text{move}(j)$ move $i$ along edge $j$; ~ (b) if $c_i=\text{jump}(j)$ move (jump) it to $node(j)$.
Any sequence (finite or infinite) of id’s of a \jag
 $J$ on an input $G$ which form consecutive legal moves of $J$ is called a \emph{computation} of $J$ on $G$. We assume that input graphs $G$ have distinguished nodes $\startnode$ and $\targetnode$, and that 
\jags have dedicated pebbles $s$ and $t$.  The initial id of $J$ on input $G$ has state $q_0$, $\node(t)=\targetnode$ and $\forall q \neq t.~ \node(q)=\startnode$.  $J$ accepts $G$ if the 
computation of $J$ on $G$ starting with the initial id ends in an id with state $q_a$. 

\paragraph{Connected components} 
A \jag cannot count, and so,  is a very weak model over discrete graphs (graphs with no edges). Since we are interested in investigating co-st-connectivity, this behaviour on graphs with a large number of components is not relevant. 
To avoid such trivialities we assume that graphs always have at most
two connected components and that each connected component is
initially pebbled with either $s$ or $t$. Thus, there is initially no
connected component without pebbles.


As already mentioned, jags are nonuniform with respect to the degree. This can be overcome by modifying the definition to place the pebbles on the edges instead of nodes, and modifying the pebble-move function to use the edge order instead of outputting a specified labeled edge. Concretely, a pebble placed on an edge can either be moved to the next edge coming out of the current edge’s source (called the \emph{next} move) or to the first edge coming out of the current edge’s target (called the \emph{first} move).

Alternatively, we can assume that a degree-$d$ graph is first transformed to a degree-$3$ graph by replacing each node with a cycle of size $d$ (replacement product). In the sequel, when running a \jag or \ndjag on a family of graphs where the degree is not known in advance,
 we always assume that one of these modifications has been put in place.

Cook and Rackoff's result \cite{jag} shows that even with this modifications, \jags can only compute local properties.

\begin{theorem}[\cite{jag}] (u)st-connectivity cannot be solved by \jags. 
\end{theorem}

\subsubsection{Counters}


Since counting is central to our study, we present various results pertaining to the ability of
 \jags to count. It is obvious that \jags
   cannot count the number of nodes  of a discrete graph.  

The following Lemma is a direct consequence of Cook-Rackoff and Reingold's result. 
\begin{lemma}\jags cannot count the number of nodes of  connected graphs.
\end{lemma}
\begin{proof}
\jags cannot explore connected graphs (\cite{jag}), while \jags augmented with counters can implement Reingold's algorithm for st-connectivity \cite{ramjag}.
\end{proof}

On the other hand, on particular graphs, \jags do have a certain
ability to count. For instance, if two pebbles are placed on nodes
connected by a (non self-intersecting) path of identically labelled (say $2$) 
edges then the
pebbles can be interpreted as a counter storing a value up to the
number of edges in this
path. 
Any implementation of counters by a deterministic \jag has to count by
repeating some fixed sequence of moves (in the above example, this was
`move along edge $2$'). The result by \cite{jag} for the deterministic
\jags used graphs with low ``exponent'' (the maximum number of times any
series of moves can be repeated without looping), so that \jags on
such graphs cannot count to a high value. This situation changes when we add non-determinism.



It is well-known that \jags can implement any  \logspace algorithm on connected graphs with a total order (made available in an appropriate way, e.g.\ by providing a path with a specific label that threads all nodes, or by an auxiliary \jag that is able to place a pebble on one node after the other).  A consequence of
the  implementability of Reingold's algorithm on 
\jags  equipped with counters is that those 
can order connected edge-labelled graphs  and so can implement any \logspace algorithm on such graphs (a run of Reingold's algorithm induces a total order on the nodes).  This shows that this model is quite powerful. 

\subsection{Nondeterministic Jumping Automata on Graphs}

\begin{definition}
 A  nondeterministic Jumping Automaton for Graphs (\ndjag)
  $J$ is a \jag  whose transition function is nondeterministic. It accepts an input if there is \emph{some} finite computation starting at the initial configuration that reaches $q_a$, and rejects an input if  \emph{no} such computation does. A $d-$\ndjag operates on graphs of degree $d$. Again, we assume that appropriate degree reduction is applied before inputting a graph to an \ndjag. 
\end{definition}
  
 It is easy to see that the argument used in \cite{jag} for deterministic 
\jags  which is similar to a pumping argument cannot be adapted easily to
\ndjags, which can solve reachability (guess a path from $\startnode$ to $\targetnode$). However, it is unclear whether \ndjags can solve co-st-connectivity. 
    Since \jags cannot count, it is reasonable to believe that \ndjags cannot implement Immerman-Szelepcsenyi's algorithm, and more generally, cannot solve co-st-connectivity.

\subsection{Cayley graphs}



Cayley graphs encode the abstract structure of groups. The Cayley graph of the group $G$ with generators $\vec{m}$, written as $CG(G, \vec{m})$ is the labelled degree $|\vec{m}|$ graph whose nodes are the elements of $G$ and where the edge labelled $i$ from node $v$ leads to $m_iv$, formally $\rho(v,i)=m_iv$. 

We take the node corresponding to $1_G$ to be $\startnode$.

We note that Cayley graphs are ``undirected'' in the relaxed sense
since for every edge there is a fixed length path (labelled by the
inverse of the generator labelling the edge) that inverts it. Indeed,
connectivity and strong connectivity coincide for them. Often, the
definition of Cayley graphs presupposes that the set of generators be
closed under inverses in which case they are undirected in the strict
sense.

It follows directly that (1) st-connectivity for Cayley graphs is in \logspace.
(2) For any two nodes $u$ and $v$, there is an automorphism that maps $u$ to $v$ 
(3) For any two pairs of nodes $(u_1,v_1)$ and $(u_2, v_2)$, either every sequence of labeled edges that reaches 
 $v_1$  from $u_1$ also reaches $v_2$ from $u_2$, or none does.  So we give paths as a list of edge-labels and define a function $\target$ 
such that $\target(u,p) = v$ iff $p$ is a path between $u$ and $v$. We refer to the list obtained by appending $b$ to $a$ by either $ab$ or $a:b$. 
(4)  For any two nodes $u$ and $v$, there is exactly one $w$ such that for every path $\rho$ from $u$ to $v$, $\rho$ is a path from $w$ to $u$.

\section{Algorithms} 



\begin{theorem}
\ndjags on Cayley graphs can  perform group multiplication and inverse.
\end{theorem}
\begin{proof} 
Multiplication: Given $p$ and $q$, to place $r$ on $node(p).node(q)$ proceed as follows: Jump pebble $p'$ to $s$ (start node) and $r$ to $q$. Nondeterministically, trace a path from $s$ to $p$ using $p'$. In tandem with the moves made by $p'$, move $r$.  Halt when  $p'$ reaches $p$.
Inverse: Given $p$ to place $q$ on a node such that $node(q).node(p) = node(s)$ proceed as follows: 
 Nondeterministically move the pebble $q$  such that $node(q).node(p) = node(s)$.
\end{proof}

\ndjags can perform other operations such as verifying that the input is \emph{not} a Cayley graph, determining if a node is an element of the center, determining if the subgroup generated by given elements is normal, etc. However, it is not obvious whether several related operations, such as checking whether the input is a Cayley graph, are doable.

If an \ndjag can be guaranteed to place a dedicated pebble on all the nodes of a class of input graph that are connected to the start node, then it follows that \ndjags can solve co-st-connectivity on this class of graphs. We call such graphs to be \emph{traversable}. However, co-\nlogspace may require a stronger property. It is well-known that if an \ndjag can order a class of graphs (i.e., can compute a total order over the nodes) then it can count, and as a consequent, decide any (co)-\nlogspace property over the class of graphs. We call such graphs \emph{orderable}.

\begin{definition} A family of graphs is \emph{traversable} if there is some \ndjag $J$ such that for each graph $G$ in the family 
all accepting computations involve placing a dedicated pebble $\curr$ on every node reachable from $\startnode$,  and $J$ does accept $G$ (i.e.\ there is at least one accepting computation). The family  is \emph{orderable} if it is traversable by an \ndjag $J$, and for every accepting computation of $J$ the ordered sequence of nodes on which $\curr$ is placed along the computation is always the same. A family of groups with generators is \emph{traversable} (or \emph{orderable}) if its Cayley graphs are. 
\end{definition}

\begin{theorem} 
 If a family of graphs is traversable,  \ndjags can solve co-st-connectivity on it. If a  family of graphs is orderable then any \ndjags can decide any \nlogspace property on the subgraphs reachable from $\startnode$.
 
\end{theorem}
\begin{proof}[sketch]
To decide co-st-connectivity simply traverse the graph and see whether $\curr$ ever reachs $\targetnode$. To decide an \nlogspace-property first note that by repeatedly cycling $\curr$ through the graph (the component reachable from $startnode$ to be precise) we can count up to its size. This allows us to simulate logarithmically sized worktapes as counters. Since the graph is not merely traversable but orderable we can define an encoding of graph nodes as numbers thus allowing us to simulate any \nlogspace Turing machine. 
\end{proof}
\subsection{Grid Graphs}

Let the family of groups $\grid(d,\ell)$ indexed by $d,l\in\mathbb{N}$ be
$(\mathbb{Z}/l\mathbb{Z})^d$. ~So $\grid(d,\ell)$ has $d$ commuting
generators (say, $m_1,\ldots,m_d$) each of order $\ell$. Cook and Rackoff
\cite{jag} showed that deterministic $d$-\jags are unable to traverse
these graphs and thus were able to conclude that deterministic \jags
cannot in general decide st-connectivity for undirected graphs.

We show that $\CG(\grid(d,\ell))$ is traversable. Every node in $\CG(\textit{\grid}(d,\ell))$ can be reached from $\startnode$ by a unique path of the form $m_1^{i_1}, : \cdots:m_d^{i_d}$ where $0 \leq  i_j < \ell$.  To ensure that $\curr$ visits every node, 
retrace nondeterministically a path $\rho$ 
of that form  to the current position $\curr$ using a helper pebble $x$. In tandem with the helper pebble moving towards $\curr$ move another helper pebble $y$ along the path that is lexicographically the successor of  $\rho$. \\
If $\rho = m_1^{i_1}: \cdots: m_d^{i_d} $  where, $(0 \leq  i_j < l) ~\wedge~ (i_k < l-1)~\wedge~ (\exists k. \forall j>k.~ i_j = l-1) $ then the lexicographically next path $\rho' =  m_1^{i_1}: \cdots: m_k^{i_k+1}$. Then jump $\curr$ to the final position of $y$. 

The following algorithm places the pebble $\curr$ on every node of the input graph $\CG(\grid(d,\ell))$.
 We also use pebbles travelling around to store directions $k$ and $dd$. 

\bigskip

\begin{verbatim}
   pebbles: s, curtrace, curr, next, count (all at startnode)
   direction: k, dd      
      
      
   repeat
       %%find the next node to move curr    
          nondeterministically guess k:{1..d}          
          for dd = 1 to k
              count := s
              nondeterministically guess b:{true,false}
              while (b)
                 nondeterministically guess b:{true,false}
                 next := next.dd
                 count := count.dd                
                 if (count == s) fail %%checks that  i_j < ell  
          if (count.k == s) fail     %%checks that  i_k < ell-1         
               
       %%check that the "immediate" next node has been found
          curtrace := next   
          for dd = k to md     %% for j>k, i_j = ell-1 
              count := s.dd  
              while (count != s) 
                 curtrace := curtrace.dd
                 count := count.dd 
          if (curtrace != curr) fail
          curr := next
   while (curr != s)
\end{verbatim}
\medskip


\begin{corollary}
 \ndjags can simulate any \nlogspace algorithm on $\CG(\grid(d,\ell))$. 
\end{corollary}
The algorithm is based on the fact that between any two nodes, there exists a unique, verifiable path $\rho$ of the form $\rho = m_1^{i_1}: \cdots: m_d^{i_d} $  where, $(0 \leq  i_j < \ell) ~\wedge~ (i_k < \ell-1)~\wedge~ (\exists k. \forall j>k.~ i_j = \ell-1) $,  and that between paths of this form, we can define a verifiable total order (the lexicographic order). Using this notion of  verifiable and orderable canonical paths, we can generalize this algorithm to other graphs. 

A predicate $R$ over paths of a graph is \emph{verifiable} if given nodes $u$, $v$, and $w$, an \ndjag can determine if there is a path $\rho$ from $u$ to $v$ with $R(\rho)$ via $w$. A total order $O$ over such a predicate $R$ over paths is \emph{verifiable} if given nodes $u$, $v$, and $w$, an \ndjag can determine if some path $\rho'$ in $R$ from $u$ is passes through $w$ and $\rho'$ is the next path of $\rho$ by $O$ where $\target(u,\rho)=v$. Given a graph $G$, define a predicate \emph{canonical path} over paths such that it is verifiable, and for any node $u$, there is exactly one canonical path from \startnode~  to $u$; and define a verifiable total order over canonical paths. 
Given any node, an \ndjag can trace the canonical path (from \startnode) to it, as well as the next canonical path. $G$ is then orderable by repeating this.

In the rest of this paper we will generalise this method first to Cayley graphs of \emph{arbitrary} abelian groups and subsequently to selected non-abelian groups.


\begin{comment}
\end{comment}


\subsection{Abelian groups}
 Here, we show that Cayley graphs of abelian groups no matter how presented can be ordered by an \ndjag. Thus, let $G$ be an abelian group 
with generators $g_1,\ldots,g_d$ . Let $e$ be the maximum order of the $g_i$. If $X\subseteq G$ we denote $\langle X\rangle$ the subgroup generated by $X$. 

\medskip

Let $\poly$ stand for a fixed but arbitrary polynomial. 
\begin{lemma} An \ndjag can count till $\poly(e)$. \end{lemma}
\begin{proof}
 By scanning through the generators, identify the one with the
maximum order, say $g_m$, and keep a pebble there. Clearly, by moving a pebble along the direction corresponding to $g_m$ one can count till $e$. Doing the same with more pebbles we can then
count till $\poly(e)$.
\end{proof}

\medskip
\begin{definition} For each $i$, let $e_i$ be the order of $g_i$ in the factor group $G / \langle g_1,\ldots,g_{i-1}\rangle$. So, $e_1$ is the order of $g_1$ in $G$, $e_2$ is the least $t$ so that $g_2^t$ can be expressed in terms of $g_1$; $e_3$ is the least $t$
so that $g_3^t$ can be expressed in terms of $g_1$ and $g_2$.
\end{definition}

A path $w = g_1^{t_1},\ldots, g_n^{t_n}$ is canonical 
if each $t_i < e_i$.

\medskip
\begin{lemma} For every element of $G$ there is a unique canonical path reaching it.\end{lemma}
\begin{proof} Clearly, every element of $G$ can be reached by a canonical path (``canonize" from $g_n$ downward). As for uniqueness suppose that
                 $g_1^{t_1},\ldots,g_n^{t_n} = g_1^{u_1},\ldots,g_n^{u_n}$
and $t_i, u_i < e_i$. Dividing by the RHS using commutativity
and reducing mod $e_i$ we obtain
                 $g_1^{t_1},\ldots,g_m^{t_m} = e$
(for different $t_i < e_i$) where $m\leq n$ and $t_m \neq 0$. This, however,
implies that $g_m^{t_m}$ can be expressed in terms of $g_1,\ldots,g_{m-1}$, a
contradiction.
\end{proof}
\newcommand{\gmax}{\textit{gmax}}
\newcommand{\nmax}{\textit{nmax}}
\begin{definition}
For each $1\leq i\leq n$ define  $\textit{gmax}_i$ as
$g_1^{e_1-1}:\cdots:g_n^{e_i-1}$
and $\textit{nmax}_i$ as the length of the canonical path leading to $\gmax_i$, i.e.,
$\nmax_i = e_1-1 +\cdots + e_i-1$
\end{definition}

\medskip
\begin{lemma}Let $w = g_1^{t_1},\ldots,g_i^{t_i}$ be a path reaching $\gmax_i$ and  of total
length $\nmax_i$, i.e., $t_1+\cdots+t_i = \nmax_i$. Then $w$ is the canonical path to $\gmax_i$ and thus $t_j = e_j-1$ for all $j$.
\end{lemma}
\begin{proof}
 By uniqueness of canonical paths the exponents $t_i$ are uniquely
determined mod $e_i$. The claim follows directly from that.
\end{proof}

\medskip
\begin{lemma} Given $\nmax_i$ and $\gmax_i$ we can reconstruct $\nmax_j$ and $\gmax_j$ for all $j<i$.
\end{lemma}
\begin{proof}
 We nondeterministically guess a path of the form $g_1^*,\ldots,g_i^*$ to $\gmax_i$ and check that its length is indeed $\nmax_i$. On its
way it passes through $\gmax_j$ for $j<i$ allowing us to read off all
intermediate values $\nmax_j$.
\end{proof}

\medskip
\begin{lemma}
 Given $\nmax_i$ and $\gmax_i$ we can enumerate the subgroup generated by $g_1,\ldots,g_i$.
 \end{lemma}
 \begin{proof}
 Using the previous lemma we can check whether a path that remains in this subgroup is canonical. This allows us to enumerate the subgroup as in the case of grid graphs described earlier.
\end{proof}

\medskip
\begin{lemma} Given $\nmax_i$ and $\gmax_i$ we can compute $\nmax_{i+1}$ and $\gmax_{i+1}$.\end{lemma}

\begin{proof}
 Compute $x_t := g_{i+1}^t$ for $t=1,2,\ldots$ and after each step check using the previous lemma determine the least $t$ so that $x_t$ is in the subgroup generated by $g_1,\ldots,g_i$. Then $\gmax_{i+1}=x_{t-1}$ and $\nmax_{i+1}=\nmax_i+t-1$.
\end{proof}

\medskip
The following is now direct.
\begin{theorem} Any abelian  group can be traversed by an \ndjag. \end{theorem}

\subsection{Families of groups}
 \textbf{Symmetric Group:} 
 Let $S(n)$ be the symmetric group of permutations of $1,\ldots,n$. This is generated by the generators $<cy,sw>$ where $sw$ is the permutation $(1,2)$ and $cy$ is the permutation $(2,3,\ldots,n,1)$.\\ 
 Canonical path: In $\CG(S(n),cy,sw)$, every node can be reached from $\startnode$ via a path of the form $p_n^{i_n},p_{n-1}^{i_{n-1}} \ldots p_2^{i_2}$ where $0< j, i_j < n $ and  $p_k = cy:(sw:cy)^{n-k}:cy^{k-1}$.  ($p_k$ is the permutation $(n-k+1,n-k+1,\ldots,n)$. Any permutation can be performed by repeating $p_1$ to place the desired element is in the first position, repeating  $p_2$ to place the desired element in the second position and so on.) 
\\

\noindent
\textbf{General linear group:}
The group of invertible matrices over finite fields with matrix multiplication as the group action. We take as generators the matrices that perform the following operations: ($\omega$)  multiplying the first column with the primitive element $\omega$ of the field; $(c_{1+2})$  adding the first column to the second; $(c_{12})$ permuting the first two columns; $(c_{cy})$ rotating the columns.

 Every invertible matrix can be transformed into a diagonal matrix with nonzero elements using row and column transformations. Thus every element in this group can be systematically  generated by generating all diagonal matrices, and for each one generated, perform each possible column transformations, and for each such transformation, perform each possible row transformations. This gives a canonical path if the transformations and diagonal matrices can be ordered. 
 
  Every column transformation (any permutation of the columns, multiplying any column with a  field element, adding the first column to another column) can be systematically generated using these generators (as in symmetric groups) using the fact that all nonzero elements of a field are generated by repeated multiplication of the primitive element ; row transformations can be simulated using column transformations (if $M$ is the operator for a particular column transformation, and the corresponding row transformation of matrix $A$ can be performed by $A.M$); diagonal matrices can be generated using repeated multiplication with the primitive element and row and column permutations. \\

\noindent
\textbf{Simple finite groups:}

Similar presentations exist for the other finite simple groups which allows their Cayley graphs to be ordered in an analogous fashion assuming a given presentation.

Furthermore, Guralnick et al \cite{guralnick} have shown that all nonabelian simple finite groups that are not \emph{Ree groups}, can be written using 2 generators and a constant number of relations. By embodying hard-coded versions of these relations it is possible to design a fixed \ndjag that can order any Cayley graph whose underlying group is one of those to which loc.cit.\ applies, irrespective of the presentation. Essentially, the \ndjag would nondeterministically guess the two generators and try out all of the hard-coded relations. Details will be given in an extended version of this paper.

Unfortunately, although any finite group can be decomposed into simple groups, it remains unclear how to use this, in order to traverse / order arbitrary Cayley graphs without prior knowledge about the underlying groups. 

\subsection{Product constructions}
\begin{theorem} If groups $G$ and $H$ are orderable, so are their direct and semidirect products. If graphs $G$ and $H$ are orderable, (and are compatible) so are their replacement product and zigzag product. Moreover, these constructions are uniform in $G,H$ thus lift to families. 
\end{theorem}
\begin{proof} The direct product of $G$ and $H$ with (generators $\vec{m}$ and $\vec{l}$ resp.), $G \times H$, has generators $\vec{m},\vec{n}$, elements of the form $(g,h)$ with $g\in G$ and $h\in H$, with $(g_1,h_1).(g_2,h_2) = (g_1.g_2, h_1.h_2)$.
To order $G\times H$, run the \ndjag that ordered $G$, and for each node $g\in G$, run the \ndjag that orders $H$ (with $s$ at $g$) to order the coset $gH$ of $H$. Semidirect product is similar.

For graphs $G$ (with edges $\vec{m}$), and $H$ (with edges $\vec{n}$) where $|H| = deg(G)$, and an order $O$ over $H$, the replacement product of $G$ and $H$, $G \textregistered H$, has vertices $(g,h)$ with $g\in G$ and $h\in H$ and edges $\vec{n}, r$ with edge $n_i$ from $(g_1,h_1)$ to $(g_1,h_2)$ if the edge $n_i$ from $h_1$ leads to $h_2$ in $H$, and edge $r$ from $(g_1,h_i)$ to $(g_2,h_j)$ if the edge $m_i$ from $g_1$ leads to $g_2$ in $G$. Given that \ndjags can order $G$ and $H$ (according to $O$), to order $G \textregistered H$, modify the \ndjag that orders $G$ so that moves along edge $m_i$ are  simulated by traversing to the $i^{th}$ node of $H$ and moving along $r$. Further, at each node of $G$, traverse all the nodes of $H$ (before following the desired edge). \footnote{Converting a family of nonuniform \ndjags over graphs $L$ parameterized by degree $d$ to a uniform \ndjag uses  $L(d)\textregistered \CG(Z/dZ)$.} Zigzag product can be traversed similarly.\\
\end{proof}
  Ordering of product groups using an \ndjag that orders its factor groups will result in an increase in the machine size. Hence, it is reasonable to assume that iterating such product constructions might not be orderable. For instance, the following construction, which iterates the wreath product construction was used in \cite{purple-reach} to construct a family of graphs over which the formal system \purple cannot solve reachability.\\
$ \Lambda(H,G,1) = H$;\\ $\Lambda(H,G,i+1) = G\wr \Lambda(H,G,i)$\\

Here, we show that if $G$ and $H$ are families of groups such that $G$ can be traversed / ordered, and $|H|$ can be computed from its number of generators
 by some register machine that does not store values greater than its output then so can the wreath product $G\wr H$ can be traversed / ordered, too. Note that we do not require that $H$ should be orderable/traversable.  This shows iterated constructions such as the above one 
(or even modifications based on diagonalisation that achieve higher growth rates, e.g.\ Ackermann function) can be traversed / ordered.

\noindent

\paragraph{Properties of wreath products.}
Consider groups $G$ and $H$ with generators $\vec{m}$ and $\vec{n}$ respectively. Their wreath product, $G \wr H$ has $|G|^{|H|}*|H|$ elements of the form $(f,h)$ where $f:H\rightarrow G$ (not necessarily a homomorphism) 
and $h\in H$, with identity $(\lambda x.id_G,id_H)$, and multiplication operation defined as $(f_1,h_1).(f_2,h_2) = (\lambda h_3.~ f_1(h_3) . f_2(h_1^{-1}. h_3), ~ h_1.h_2)$.

Both $G$ and $H$ embed faithfully into $G\wr H$ by $g\mapsto
(\delta_g,1_H)$ and $h\mapsto (\delta_{1_G},h)$ where
$\delta_g(1_H)=g$ and $\delta_g(h)=1_G$ when $h\neq 1_H$. We identify
elements of $G$ and $H$ with these embeddings where appropriate. The
(embeddings of the) generators $G$ and $H$ together generate $G\wr H$
and we define those to be the distinguished generators of the wreath
product.

If $\vec h=h_1,\dots,h_n$ are pairwise distinct elements from $H$ and
$\vec g = g_1,\dots g_n$ are (not necessarily distinct) elements from
$G\setminus\{1_G\}$ we define $[\vec h{\mapsto}\vec g]=(f,1_H)\in G\wr
H$ where $f(h_i)=g_i$ and $f(h)=1_G$ when $h\not\in\{h_1,\dots,h_n\}$.
 
We have
\begin{equation}\label{edrt}
[\vec h{\mapsto}\vec g] = h_1g_1h_1^{-1}h_2g_2h_2^{-1}\dots h_ng_nh_n^{-1}
\end{equation}

\paragraph{Counting up to $|H|$.} We can now use elements of this form to represent numbers; to be precise we represent $n$ by elements $(f,1_H)$ where $|\{h\mid f(h)\neq 1_G\}|=n$. 

We now consider $\ndjags$ run on the Cayley graph of 
$G\wr H$ with the above generators. We assume a fixed start pebble allowing us to identify pebble positions with elements of $G\wr H$. 

\medskip
\begin{lemma}
We can test  whether an element $x\in G\wr H$  represents a number.
\end{lemma}
\begin{proof}
Nondeterministically guess a path (from the fixed start pebble)
  to $x$ and check that the $H$-generators contained in that path
 multiply to $1_H$. 
\end{proof}

\medskip
\begin{lemma}\label{testf}
Given a number representation $x=(f,1_H)$ and an element $h\in H$ (both by pebbles) it is possible to test whether $f(h)=1_G$. 
\end{lemma}
\begin{proof}
Nondeterministically move a pebble $a$ to $x$ on a path $\rho$. Use a distinguished pebble $u$ to store the current $H$-value (the second component of $a$) as we go along the path. To do this, we apply the $H$-generators contained in the path $\rho$ to $u$. Another distinguished pebble $v$ will represent the value $f(h)\in G$ when $a=(f,u)$. To do this, check at all times whether $u$ happens to be equal to $h$ and in that case apply $G$-generators in $\rho$ to $v$. When $u\neq h$, let those $G$-generators pass. 
Then, when $a$ finally reaches $x$, we check whether $v=1_G$. 
\end{proof}

\medskip
\begin{lemma} \label{canpath}
Given a number representation $x$ it is possible to nondeterministically guess a path to $x$ that has the format in Equation~\ref{edrt}. 
\end{lemma}
\begin{proof}
We trace a path $\rho$  to $x$ and store the values $h_1$ and $h_i$ in pebbles $u$ and $v$. Initially we set both $u$ and $v$ to the first $H$-block (maximal subword of $H$-generators) in $\rho$. As we progress, we accumulate the next $H$-block in a pebble $w$, compute $h:=vw$ and check using Lemma~\ref{testf} that $h$ is not among $h_1,\dots,h_i$. Should this happen, we fail. Otherwise, we can update $v$ to $h$ and continue. 
\end{proof}

\medskip
\begin{lemma}
Given two number representations we can test whether they represent the same number and whether one represents the successor of the other. 
\end{lemma}
\begin{proof}
This is a direct consequence of Lemma \ref{canpath}: trace canonical paths and check that their respective numbers of $H$-blocks are equal or in the successor relation. 
\end{proof}

\medskip
\begin{lemma}\label{ndcount}
If $|H|$  is computable from the number $|\vec n|$ of $H$'s generators
 by a register machine which never stores a value larger than its final output 
  then an \ndjag can count till $|H|$.
\end{lemma}
\begin{proof}
The control of the register machine can be hard-coded in the transition table of an \ndjag while number representations are used to encode registers. Register values are guaranteed never to exceed $|H|$ and the operations of increment, decrement and checking for zero can be done as described in the above lemmas.
\end{proof}

For example, if $|H|$ is an iterated power of $2$ then the premise of the lemma applies. The lemma is needed in the first place because the simulation of arithmetic contained in the earlier lemmas does not in general allow the detection of overflow and hence wraparound. Once we know $|H|$ (e.g.\ by the simulation of a register machine which computes $|H|$ from $|\vec n|$,  never storing a value larger than $|H|$) we can then count till $|H|$. 
\begin{theorem}
Let $G$, $H$ be groups such that $G$ can be traversed(ordered) and such that the generators of $H$ are either given explicitly or can be computed by a single $\ndjag$. Then there exists an \ndjag traversing(ordering) the Cayley graph of $G\wr H$ and the construction is uniform in $G, H$.
\end{theorem}
\begin{proof}
The previous lemmas together with Theorem~\ref{ndcount} show that we can perform arbitrary arithmetic operations with values up to $\textit{poly}(|H|)$. This allows us to implement Reingold's algorithm on $H(n)$ and hence traverse (order) $H$. It is then easy to traverse (order) $G \wr H$. 
\end{proof}

\section{Conclusion}
 Algorithms for co-st-connectivity use techniques based on counting. We wanted to explore the question of whether the operation of counting is necessary to solve the problem of co-st-connectivity. For this purpose we considered jumping automata on graphs  systems which cannot count, and consequently cannot do u-st-connectivity, but are powerful otherwise. This was the weakest of such models. 
 
 We added nondeterminism to \jags, and studied the problem of coreachability on Cayley graphs which were used to show that deterministic \jags cannot solve undirected reachability. This result exploited the rigid structure of such graphs - in particular, the fact that for any two nodes of a Cayley graphs, there is an automorphism that maps one to the other. To our surprise, we found that nondeterminism can exploit this property. Many families of Cayley graphs and product constructions over them can be ordered using \ndjags,  and in consequence any (co)-\nlogspace algorithm can be implemented, including co-st-connectivity. Since this machine model is the weakest, these results hold for stronger models such as positive transitive closure logic as well. 
 

While we are as yet unable to show that \ndjags can traverse arbitrary graphs, our results considerably restrict the search space for possible counterexamples. On the one hand, they should exhibit a very regular structure so that deterministic \jags cannot already solve reachability; on the other hand, the very presence of such regular structure allowed the applicability of our new nondeterministic algorithms in all concrete cases.


\bibliography{bib}
\end{document}